\documentclass[review]{elsarticle}

\usepackage{hyperref}
\usepackage{moreverb,url}

\usepackage{mathrsfs,fullpage,bm}
\usepackage{dsfont}
\usepackage{float}
\usepackage{bbm}
\usepackage{amsmath}
\usepackage{graphicx}
\usepackage{amssymb}
\usepackage{amsthm}
\usepackage{verbatim}
\usepackage{epsfig}
\usepackage[labelfont=bf]{caption}
\usepackage{tabularx}
\usepackage{enumerate}
\usepackage{circuitikz}
\usetikzlibrary{calc} 
\usepackage{subfig}
\usepackage{tabularx}
\usepackage{multirow}
\usepackage{multicol}%
\usepackage{booktabs} %
\usepackage{epstopdf}
\makeatletter
\newcommand{\pushright}[1]{\ifmeasuring@#1\else\omit\hfill$\displaystyle#1$\fi\ignorespaces}
\makeatother
\newtheorem{fact}{Fact}
\newtheorem{definition}{Definition}
\newtheorem{assumption}{Assumption}
\newtheorem{lemma}{Lemma}
\newtheorem{remark}{Remark}

\newtheorem{corollary}{Corollary}
\newtheorem{proposition}{Proposition}
\def\begquo{\begin{quote}}
	\def\endquo{\end{quote}}
\def\begequarr{\begin{eqnarray}}
	\def\endequarr{\end{eqnarray}}
\def\begequarrs{\begin{eqnarray*}}
	\def\endequarrs{\end{eqnarray*}}
\def\begarr{\begin{array}}
	\def\endarr{\end{array}}
\def\begequ{\begin{equation}}
	\def\endequ{\end{equation}}
\def\lab{\label}
\def\begdes{\begin{description}}
	\def\enddes{\end{description}}
\def\begenu{\begin{enumerate}}
	\def\begite{\begin{itemize}}
		\def\endite{\end{itemize}}
	\def\endenu{\end{enumerate}}

\def\lef[{\left[\begin{array}}
	\def\rig]{\end{array}\right]}

\def\begcen{\begin{center}}
	\def\endcen{\end{center}}
\def\begrem{\begin{remark}\rm}
	\def\endrem{\end{remark}}
\def\begdef{\begin{definition}}
	\def\enddef{\end{definition}}
\def\begpro{\begin{proposition}}
	\def\endpro{\end{proposition}}
\def\begfac{\begin{fact}}
	\def\endfac{\end{fact}}
\def\begass{\begin{assumption}}
	\def\endass{\end{assumption}}

\def\begsubequ{\begin{subequations}}
	\def\endsubequ{\end{subequations}}
\def\begmat#1{\begin{bmatrix}#1\end{bmatrix}}
\def\begali#1{\begin{align}{#1}\end{align}}
\def\begalis#1{\begin{align*}{#1}\end{align*}}
\def\calg{{\cal G}}

\def\calo{{\cal O}}

\def\calc{{\cal C}}

\def\calw{{\cal W}}

\def\cald{{\cal D}}

\def\L2e{{\cal L}_{2e}}

\def\rea{\mathds{R}}

\def\adj{\mbox{adj}}
\def\col{\mbox{col}}

\def\TAC{{\it IEEE Trans. Automatic Control}}

\def\SCL{{\it Systems and Control Letters}}
\def\AUT{{\it Automatica}}

\def\calg{{\cal G}}

\def\calo{{\cal O}}

\usepackage{color}

\def\calw{{\cal W}}

\usepackage[prependcaption,colorinlistoftodos]{todonotes}

\def\calp{{\mathfrak p}}
\journal{Journal of \LaTeX\ Templates}

\begin{document}
\begin{frontmatter}
\title{\textbf{On-line Estimation of the Parameters of the Windmill Power Coefficient }}
\author[itmo]{Alexey Bobtsov}
\author[itam,itmo]{Romeo Ortega}
\author[supelec,itmo]{Stanislav Aranovskiy}
\author[itam]{Rafael Cisneros\corref{mycorrespondingauthor}}
\cortext[mycorrespondingauthor]{Corresponding author}
\ead{rcisneros@itam.mx}

\address[itmo]{ ITMO University, 197101 St. Petersburg, Rusia}
\address[itam]{Instituto Tecnol\'ogico Aut\'onomo de M\'exico, 01080 Mexico city, Mexico}
\address[supelec]{CentraleSup\'elec, 35576 Rennes, France}

%%%%%%%%%%%%%%%%
\begin{abstract}
	Wind turbines are often controlled to harvest the maximum power from the wind, which corresponds to the operation at the top of the bell-shaped power coefficient graph. Such a mode of operation may be achieved implementing an extremum seeking data-based strategy, which is an invasive technique that requires the injection of harmonic disturbances. Another approach is based on the knowledge of the analytic expression of the power coefficient function, an information usually unreliably provided by the turbine manufacturer. In this paper we propose a globally, exponentially convergent  on-line estimator of the parameters entering into the windmill power coefficient function. This corresponds to the solution of an identification problem for a nonlinear, nonlinearly parameterized, underexcited system.  To the best of our knowledge we have provided the first solution to this challenging, practically important, problem. 
\end{abstract}
%%%%%%%%%%%%%%%%
%

\end{frontmatter}

%\linenumbers
%%%%%%%%%%%%%%%%%%5
\section{Introduction}
\lab{sec1}
%%%%%%%%%%%%%%%%%%
%
Wind power is a primary renewable energy source. For the majority of the time wind turbines operate below the rated wind speed; the so-called Region-2 operation. The primary objective of Region-2 wind turbine control is to harvest the maximum power, a goal that is referred to as maximum power point tracking (MPPT). To achieve this objective it is necessary to operate the wind turbine at the maximum point of the power map, which is a bell-shaped function---usually called the $\calc_p$ power coefficient---that has a unique  maximum power point with respect to turbine speed at each level of wind speed. Given the knowledge of the wind speed the rotor speed is then selected to operate the turbine at this point.

It is clear that to achieve MPPT it is necessary to compute the maximum of the windmill power coefficient. MPPT algorithms are divided into two groups \cite{KUMCHA}. One uses  {\em data-driven} approaches, {\em e.g.}, extremum seeking \cite{GHAetal}, based on the electrical current and voltage measurements. This is an invasive procedure that requires the injection of harmonic perturbations to the turbine, which is undesirable. The second one uses the  wind measurements and assumes the {\em knowledge} of the $\calc_p$ curve. If this knowledge is accurate, the latter are more efficient and faster in their response. The $\calc_p$ curve is usually approximated by analytic nonlinear functions, as first suggested in \cite{WASetal}---see also \cite[Subsection 2.5.2]{HEIbook}. To faithfully capture the bell-shaped form of this curve the functions are the product of a polynomial expression and an exponential one. The coefficients on this function are then selected to match the data fields collected for specific turbines---a step that is critical to achieve the control objective.

In this paper we propose an identification procedure to {\em estimate on-line} the coefficients of the classical $\calc_p$ function given measurements of the rotor angular speed. We consider two possible scenarios:\\

\noindent {\bf S1}  The wind speed is {\em constant} and known and the windmill is operated {\em off-grid}.\\

\noindent {\bf S2}  The wind speed is not constant, its {\em derivative is known} and it is possible to manipulate the electrical torque.\\

 Given the availability of various wind-speed estimators, the assumption that the wind speed is known is very mild.\footnote{A recent comparative experimental study of various wind speed estimators may be found in \cite{SOLetal}.}  On the other hand, the additional requirement that the wind speed is {\em constant} is rather restrictive. However, as shown in the simulations of Section \ref{sec8}, the estimator is robust {\em vis-\`a-vis} variations of the wind speed. The assumption of off-grid operation stymies the application of our estimator in an on-line adaptive MPPT. This obstacle is partially relaxed in the scenario {\bf S2}, but it has the drawback that we need to know the wind's acceleration. Moreover, as shown below, it is necessary to impose a particular behavior to the generator torque. Although this is, in principle possible manipulating the converter current, this is not a standard procedure in applications. In view of this situation we will concentrate in the paper in the scenario {\bf S1} and just indicate the key steps required to implement the estimator in the scenario {\bf S2}. \\

It is shown in the paper that, the problem of estimation of the coefficients of the $\calc_p$ curve, boils down to the solution of the following.\\ 

\noindent {\bf Identification Problem} Given a scalar, nonlinear system described by the differential equation 
$$
\dot z=-z^3(\theta_1 z- \theta_2)e^{-\theta_3 z}, 
$$
with $z(t)>0$ measurable and the constant coefficients $\theta_i>0$ such that the system has a unique globally asymptotically stable equilibrium point $\bar z> 0$. Construct an on-line parameter estimator that ensures $\hat \theta(t) \to \theta$, globally and exponentially for all initial conditions $z(0)>0$.\\

To the best of our knowledge none of the existing identification techniques provides an answer to this problem, which is particularly difficult because: (i) the system dynamics is highly nonlinear; (ii) the nonlinear dependence on the unknown parameters; (iii) the extremely ``weak excitation" conditions of the system---namely, a trajectory $z(t)$ starting in some initial condition and converging to a constant equilibrium point. 

We provide in this paper the first solution to this problem proceeding as follows. First, we introduce a suitable (nonlinear) dynamic extension that allows us to obtain a regression equation for the system. Unfortunately, this is a {\em nonlinearly parameterized regressor equation} (NLPRE), that is, a relation of the form $y=\phi \calg(\theta)$, where the mapping $\calg:\rea^3 \to \rea^4$ is {\em known} and the signals $y(t) \in \rea^2$ and $\phi(t) \in \rea^{2 \times 4}$ are {\em measurable}. The first observation is that {\em overparameterizing} to obtain a linear regression equation (LRE), is not effective because of the lack of sufficient excitation in the regressor. On the other hand, there are parameter estimators for NLPRE of this form, {\em e.g.}, \cite{ORTetalaut21,WANetal}, however they impose a monotonicity condition on the mapping $\calg$, that is not verified in our example. A second step is then to construct a new NLPRE that verifies the monotonicity requirement. Alas, this is still not sufficient for the application of the estimators reported in  \cite{ORTetalaut21,WANetal}---again because of the weak excitation conditions. Therefore, the third step is to {\em extend} the least squares (LS)-based estimator recently introduced in \cite{PYRetal} for LRE to make it applicable to NLPRE. This estimator is a combination of LS and  the dynamic regressor extension and mixing (DREM) algorithm \cite{ARAetaltac17,ORTetaltac21}, hence it is referred in sequel as  [LS+DREM] estimator. The two main features of this estimator is that it ensures global and exponential convergence under the weak assumption of {\em interval excitation} (IE) \cite{KRERIE,TAObook} of the regressor. Moreover, compared with gradient-based estimators, LS ones enjoy  superior convergence properties and noise insensitivity, facts which are widely recognized in the identification literature \cite{LJUbook}.

The remainder of the paper is organized as follows. In Section \ref{sec2} we present the analytic expression of the $\calc_p$ function studied in the paper.  In Section \ref{sec3} we formulate the problem to be solved. Section \ref{sec4} contains some preliminary lemmata. In Section \ref{sec5} we derive the first NPLRE of the system dynamics and in Section \ref{sec6} we construct a new parameterization that satisfies the monotonicity requirement needed for the estimator. Section \ref{sec7} contains our main  result, namely the LS+DREM estimator of the parameters of the power coefficient. Simulation results, which illustrate the {performance} of the proposed estimator, are presented in Section \ref{sec8}. The paper is wrapped-up with concluding remarks and future research in  Section \ref{sec9}. \\

\noindent {\bf Notation.} $I_n$ is the $n \times n$ identity matrix.   For $a \in \rea^n$, we denote $|a|^2:=a^\top a$.  We define the derivative operator $\calp[x]:=\frac{dx(t)}{dt}$. The action of an LTI filter  ${\mathfrak F}(\calp) \in \rea(\calp)$ on a signal $x(t)$ is denoted as ${\mathfrak F}[x]$.  To simplify the notation, whenever clear from the context, the arguments of the various functions are omitted. 
%
%%%%%%%%%%%%%
\section{Power Coefficient of the Windmill}
\lab{sec2}
%%%%%%%%%%%%%
%
In this section we present the mathematical model of the power coefficient of the wind turbine that we will consider in the paper. For further details on modeling of windmill systems the reader is referred to the classical textbook  \cite{HEIbook}.  

The power captured by the windmill's turbine is given by the expression
$$
P_w=\kappa v_w^3\calc_p\left(\lambda,\beta\right),
$$
where $\kappa:=\frac{1}{2}\rho A$, $\rho$ is the air density, $A$ is the area swept by the blades, $v_w > 0$ is the wind speed and $\calc_p$ is the so-called {\em power coefficient} that depends on the turbine's pitch angle $\beta \geq 0$ and the tip-speed ratio $\lambda$ defined as
\begin{align}\label{lambda}
\lambda:=\frac{r\omega}{v_w}
\end{align}
where $\omega>0$ is the {\em rotor's angular speed} and $r$ is the {blade's length}. 

In  \cite[Section 2.1.5.2]{HEIbook}, the power coefficient is defined via the function
\begin{align}\label{Cp}
	\calc_p \left(\lambda,\beta\right) = \kappa_1 \Big(\frac{\kappa_2}{\lambda_i} - \kappa_3 \beta - \kappa_4 \beta^\ell-\kappa_5 \Big)e^{-\frac{\kappa_6}{\lambda_i}}
\end{align}
with 
\begequ
\lab{lami}
\frac{1}{\lambda_i}:=\frac{1}{\lambda+0.08\beta}-\frac{\kappa_7}{\beta^3+1},
\endequ
where $\kappa_i>0,\;i=1,\dots,7,$ and $\ell$ are tuning coefficients for the curve fitting.   
  
As indicated in the Introduction, in this paper we  consider the scenario when the wind speed is below the rated speed of the turbine. Then $\beta$ is set to zero to get the maximum power extraction. Thus, setting $\beta=0$  and substituting \eqref{lami} into  \eqref{Cp} yields
\begin{equation}\label{heier}
\calc_p \left(\lambda,0\right)= \kappa_1 \Big[\kappa_2\Big(\frac{1}{\lambda}-\kappa_7\Big) -\kappa_5\Big]e^{-\kappa_6\big(\frac{1}{\lambda}-\kappa_7\big)}.
\end{equation}
To simplify the notation in the sequel we substitute \eqref{lambda} into the equation above, define the new signal
\begequ
\lab{z}
z:=\frac{v_w}{\omega},
\endequ
and, after some manipulations,  we get the following compact expression of the power coefficient
\begin{equation}
\label{cp}
	C_p(z):=\calc_p \left(\frac{r}{z},0\right)=c_1(z-c_2)e^{-c_3z},
\end{equation}
where we defined the positive constants
$$
	c_1:= \frac{1}{r}\kappa_1\kappa_2 e^{\kappa_6\kappa_7},\;c_2:=r(\kappa_7+\frac{\kappa_5}{\kappa_2}), \;c_3:=\frac{\kappa_6}{r}.
$$
The  shape of the function $C_p(z)$ is depicted in Fig. \ref{Cpfig}.
\begin{figure}
	\centering
\includegraphics[scale=0.3]{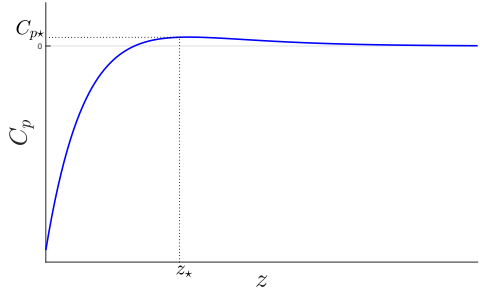}
	\caption{Power coefficient $C_p \left(z\right)$. }%\rom{This fig is not clear. If it is $C_p(z)$ the argument $z$, as well as the point $z_*$, should appear on the horizontal axis of the plot. No need to mention $\lambda$ anymore}}
	\label{Cpfig}
\end{figure}
 
As indicated in the Introduction, customary control application consists in making the turbine operate at the speed in which the maximum extraction of power is attained---a strategy called MPPT. Equivalently stated, the turbine must operate at the value of $z=z_{\star}$ of Fig. \ref{Cpfig}. 
%\todo[inline, size = \small]{Sorry, but I do not see $z_{\star}$ in the figure.}
Clearly, the value of $z_{\star}$ is determined by the values of the parameters $c_i$ in \eqref{cp},
\[
	z_{\star} = \frac{c_2 c_3 + 1}{c_3}.
\]
Thus, the knowledge of these parameters is crucial for the optimal operation of the wind turbine.
%
%%%%%%%%%%%%%%5
\section{Dynamics of the System and Problem Formulation}
\lab{sec3}
%%%%%%%%%%%%%%%5
%
In this work we consider the one-mass model of the wind turbine that yields the mechanical dynamics 
\begin{equation}
\lab{mecdyn}
J\dot \omega = T_m - T_e
\end{equation}
where $J>0$ is the rotor inertia and $T_m$, $T_e$ are the mechanical and electrical torques, respectively. The mechanical torque is given by
\begin{equation}
\lab{tm}
T_m= \kappa \frac{v_w^3}{\omega} C_p\Big(\frac{v_w}{\omega}\Big),
\end{equation} 
where the {power coefficient}  $C_p$ is defined in \eqref{cp}.
\subsection{Parameter Estimation Task} 

We proceed now to formulate the estimation problem for the two scenarios {\bf S1} and {\bf S2} described in the Introduction. In both cases we impose the following.\\

\noindent {\bf Assumption A1} [Rotor speed measurement]  The signal $\omega$ is bounded and {\em measurable}.\\

The objective is to design for the system \eqref{mecdyn}, \eqref{tm} an  estimator of the power coefficient parameters $c:=\col(c_1,c_2,c_3)$ of the form
\begali{
\nonumber
\dot {x}_c & =f_c(x_c, \omega)\\
\lab{paa}
\hat c & =h_c(x_c,\omega),
}
where $f_c: \rea^{n_c} \times \rea_+  \to \rea^{n_c},\;h_c:  \rea^{n_c} \times\rea_+  \to \rea^3,$ which ensures {\em global, exponential convergence} of the parameter errors $\tilde c:=\hat c-c$ to zero, with all signals bounded.
\subsubsection*{Additional assumptions for Scenario {\bf S1}}  

\noindent {\bf Assumption  A2} [Wind]  The wind speed $v_w$ is a {\em known constant}.\\

\noindent {\bf Assumption  A3} [Off-grid operation]  The windmill is operating off-grid with {\em no electrical torque}, whence $T_e=0$. 
\subsubsection*{Additional assumptions for Scenario {\bf S2}}  
\noindent {\bf Assumption  A2'} [Wind]  The wind speed $v_w$ is time-varying and, together with its derivative $\dot v_w$, they are {\em measurable}.\\

\noindent {\bf Assumption  A3'} [Generator toque]  The generator torque is chosen to satisfy
\begequ
\lab{refte}
T_e=-J\frac{\dot v_w}{v_w}\omega.
\endequ
\subsection{Discussion on the assumptions}  
\noindent {\bf D1}  The assumption of known rotor speed is quite natural and, actually, necessary for the solution of the estimation problem.\\

\noindent {\bf D2}  The assumption of known wind speed is very mild. The additional requirement that the wind speed is {\em constant} is, on the other hand, rather restrictive. However, as shown in the simulations of Section \ref{sec8}, the LS-based estimator is robust {\em vis-\`a-vis} variations of the wind speed.\\

\noindent {\bf D3} As indicated above the standing assumption in scenario {\bf S1} is that the windmill is operating in an {\em off-grid} scenario. In that case the electrical torque $T_e$ is zero. As will become clear below the problem of parameter estimation with $T_e \neq 0$ is much more complicated and, at this point, we are unable to solve it. Unfortunately, this scenario rules out the possibility of doing on-line estimation of the power coefficient parameters of a windmill transferring power to the grid, which might lead to a practically interesting  model-based adaptive MPPT  controller.\\

\noindent {\bf D4} Interestingly the problem mentioned above regarding the undesired impact of the generator torque {\em disappears} if we can impose to it the value given by \eqref{refte}. As will become clear below, this choice is justified only by the desire to cancel the effect of the term $T_e$ in the regressor equation needed for the parameter estimation---and can hardly be motivated by any other practical consideration. Therefore, in the sequel we concentrate out attention on scenario {\bf S1} and just indicate how to proceed with the design in scenario {\bf S2}.
%
%%%%%%%%%%%%%%5
\section{Preliminary Lemmata}
\lab{sec4}
%%%%%%%%%%%%%%%5
%
In this section we present three preliminary results which are instrumental for the solution of the problem.
\subsection{A change of coordinates}
\lab{subsec41}
%%%%%%%%%%%%%%%5
%
\begin{lemma}\em
\lab{lem1}
The dynamics  \eqref{mecdyn},  \eqref{tm}---verifying {\bf Assumption  A2}---expressed in the coordinate $z$ defined in \eqref{z}  takes the form
\begequ
\lab{zdyn}
\dot z =-z^3(\theta_1 z- \theta_2)e^{-\theta_3 z} + z^2  \tau,
\endequ
where,  to simplify the notation, we defined the signal
$$
	\tau:=\frac{T_e}{Jv_w},
$$
and the constant vector
\begequ
\lab{the}
\theta=\begmat{\theta_1 \\ \theta_2\\ \theta_3}:=\begmat{\frac{\kappa v_w}{J} c_1 \\ \frac{\kappa v_w}{J} c_1c_2 \\ c_3},
\endequ
whose inverse is given as
\begequ
\lab{ccalg}
c=\begmat{c_1 \\ c_2\\ c_3}=\begmat{\frac{J}{\kappa v_w} \theta_1 \\ \frac{\theta_2}{\theta_1}\\ \theta_3}=:\calc(\theta).
\endequ
\end{lemma}
\begin{proof}
Consider the following calculations
\begalis{
\dot z &=-\frac{v_w}{\omega^2}\dot \omega\\
&= \frac{v_w}{\omega^2}\Big( \frac{T_e}{J}-\frac{\kappa}{J} \frac{v_w^3}{\omega} C_p\Big(\frac{v_w}{\omega}\Big) \Big)\\
&= \tau z^2-z^3\frac{\kappa v_w}{J}  C_p(z) \\
&= \tau z^2-z^3\frac{\kappa v_w}{J} c_1(z-c_2)e^{-c_3 z} \\
&= \tau z^2-z^3\Big(\frac{\kappa v_w}{J} c_1 z-\frac{\kappa v_w}{J} c_1c_2 \Big)e^{-c_3 z}.
}
The proof is completed recalling \eqref{the}.
\end{proof} 
The following corollary is a slight variation of the lemma above, which is of interest for the scenario {\bf S2}.

\begin{corollary}\em
\lab{cor1}
The dynamics  \eqref{mecdyn},  \eqref{tm}---verifying {\bf Assumption  A2'} and {\bf Assumption  A3'}---expressed in the coordinate $z$ defined in \eqref{z}  takes the form
\begequ
\lab{zdyn1}
\dot z = -v_w z^3\Big(\bar \theta_1 z- \bar \theta_2\Big)e^{-\bar \theta_3 z},
\endequ
where the constant vector
\begequ
\lab{the1}
\bar \theta=\begmat{\bar \theta_1 \\ \bar \theta_2\\ \bar \theta_3}:=\begmat{{\kappa \over J} c_1 \\ {\kappa  \over J} c_1c_2 \\ c_3},
\endequ
whose inverse is given as
$$
c=\begmat{c_1 \\ c_2\\ c_3}=\begmat{{J\over \kappa } \bar \theta_1 \\ {\bar \theta_2 \over \bar \theta_1}\\ \bar \theta_3}=:\bar \calc(\bar \theta).
$$
\end{corollary}
\begin{proof}
Consider the following calculations
\begalis{
\dot z &=-\frac{v_w}{\omega^2}\dot \omega + \frac{\dot{v}_w}{\omega}\\
&= \frac{v_w}{\omega^2}\Big( {T_e \over J}-{\kappa \over J} \frac{v_w^3}{\omega} C_p\Big(\frac{v_w}{\omega}\Big) \Big)+ \frac{1}{v_w}z\dot{v}_w\\
&= \frac{z^2}{Jv_w} T_e-z^3{\kappa v_w\over J}  C_p(z)+ \frac{1}{v_w}z\dot{v}_w \\
&= \frac{z^2}{Jv_w} T_e-z^3{\kappa v_w\over J} c_1(z-c_2)e^{-c_3 z}+ \frac{1}{v_w}z\dot{v}_w \\
&= -z^3v_w\Big({\kappa \over J} c_1 z-{\kappa \over J} c_1c_2 \Big)e^{-c_3 z},
}
where we used \eqref{refte} to obtain the last identity. The proof is completed recalling \eqref{the1}.
\end{proof} 
\subsection{Discussion} 
\lab{subsec42} 

Before proceeding with our derivations we take a brief respite to prepare the coming discussion. As is well known a key step for the solution of any parameter identification problem is the derivation of a regression equation. Unfortunately, the technique that we use to do it here is applicable to \eqref{zdyn} only if the additive term $z^2 \tau$ is {\em absent}. Hence the need of {\bf Assumption A3}---that is, assume $T_e=0$---in the scenario {\bf S1}. Interestingly, as shown in Corollary \ref{cor1}, this additive term is absent in \eqref{zdyn1}, but to obtain this result it was necessary to impose the technically motivated {\bf Assumption A3'}.  

There is a second difference between the dynamical systems \eqref{zdyn} and \eqref{zdyn1}, namely, the presence of the signal $v_w$ scaling $\dot z$ in  \eqref{zdyn1}. As will become clear below, this difference is of minor consequences, and just modifies slightly the design of the identifier. 
\subsection{Two key identities}
\lab{subsec43}
%%%%%%%%%%%%%%%5
%
The construction of the regressor equation carried out in Section \ref{sec5} relies on the following key identities. We present it in detail for the case of the dynamics \eqref{zdyn} in the lemma below and just sketch it as a corollary for the case of \eqref{zdyn1}.

\begin{lemma}
\lab{lem2}\em
Consider the representation of the system dynamics \eqref{mecdyn},  \eqref{tm} given in \eqref{zdyn}. Define the signals
	\begali{
	\nonumber
		\dot{\xi}_1&=-z^4,\;\xi_1(0)=0\\
		\dot{\xi}_2&=z^3,\;\;\;\xi_2(0)=0.
\lab{dotxi}
}
The following relation holds
\begequ
\lab{keyide0}		
e^{\theta_3 z(0)}\dot z  +\theta_1\theta_3\xi_1\dot z + \theta_2 \theta_3\xi_2\dot z +\theta_1z^4-\theta_2z^3=\tau_d,
\endequ
where we defined the signal
\begali{
\lab{taud}
\tau_d&:=\tau e^{\theta_3 z}z^2- \theta_3 \dot z  \int_0^t \tau(s) e^{\theta_3 z(s)}z^2(s) ds.
 }
\end{lemma}

\begin{proof}
	Multiplying both sides of \eqref{zdyn} by $e^{\theta_3 z}$ yields
	\begin{equation}
	\lab{ethe3}
		e^{\theta_3 z}\dot z= -\theta_1z^4 +\theta_2 z^3+\tau e^{\theta_3 z}z^2.
	\end{equation}
Defining the signal
	\begequ
	\lab{q}
	q:=\frac{1}{\theta_3}e^{\theta_3 z}
	\endequ
	we get
	\begalis{
	\dot{q}&=  e^{\theta_3 z}\dot z\\
	&=-\theta_1 z^4 +  \theta_2 z^3+\tau e^{\theta_3 z}z^2,
	}
where we used \eqref{ethe3} in the second identity. Replacing the signals defined in \eqref{dotxi}  in the equation above yields
$$
\dot{q}=\theta_1 \dot{\xi}_1 +  \theta_2 \dot{\xi}_2+\tau e^{\theta_3 z}z^2.
$$
Integrating this equation, using \eqref{q} and multiplying by $\theta_3$, we get
\begequ
\lab{ethe3z}
e^{\theta_3z}=e^{\theta_3 z(0)}+\theta_1 \theta_3{\xi}_1 +  \theta_2 \theta_3{\xi}_2+\theta_3 \int_0^t \tau(s) e^{\theta_3 z(s)}z^2(s) ds.
\endequ
Multiplying both sides of \eqref{ethe3z} by $\dot z$, and using \eqref{ethe3}, we get
\begalis{
	 \dot z e^{\theta_3 z}&= \dot z [ e^{\theta_3 z(0)} +\theta_1\theta_3\xi_1+ \theta_2 \theta_3\xi_2+ \theta_3 \int_0^t \tau(s) e^{\theta_3 z(s)}z^2(s) ds.]\\
&=-\theta_1z^4+\theta_2z^3+\tau e^{\theta_3 z}z^2,
}
completing the proof of the claim.
\end{proof}

The corollary below pertains to the scenario {\bf S2}, that is, to the system dynamics \eqref{zdyn1}. Since its proof follows {\em mutatis mutandi} the proof of Lemma \ref{lem2}, we only sketch it briefly. 

\begin{corollary}
\lab{cor2}\em
Consider the representation of the system dynamics \eqref{mecdyn},  \eqref{tm} given in \eqref{zdyn1}. Define the signals
	\begalis{
		\dot{\bar \xi}_1&=-v_w z^4,\;\bar \xi_1(0)=0\\
		\dot{\bar \xi}_2&=v_w z^3,\;\;\;\bar \xi_2(0)=0.
}
The following relation holds
\begequ
\lab{keyide1}		
e^{\bar \theta_3 z(0)}\dot z  +\bar \theta_1 \bar \theta_3\bar \xi_1\dot z + \bar \theta_2 \bar \theta_3 \bar \xi_2\dot z +\bar \theta_1 v_w z^4-\bar \theta_2 v_w z^3=0.
\endequ
\end{corollary}

\begin{proof} (Sketch)
	Multiplying both sides of \eqref{zdyn1} by $e^{\bar \theta_3 z}$ yields
	$$
		e^{\bar \theta_3 z}\dot z= -\bar \theta_1v_w  z^4 +\bar \theta_2 v_w z^3.
	$$
The proof is completed mimicking the derivations in the proof of Lemma \ref{lem3} but using the signal
$$
	\bar q:=\frac{1}{\bar \theta_3}e^{\bar \theta_3 z},
$$
instead of $q$ defined in \eqref{q}.
\end{proof}
\subsection{The Swapping lemma}
\lab{subsec44}
%%%%%%%%%%%%%%%5
%
Instrumental for the establishment of our result is the well-known Swapping Lemma  \cite[Lemma 6.3.5]{SASBODbook} that we present below in the form that is used here.

\begin{lemma}
\lab{lem3}\em
Given the differentiable scalar signals $x$ and $u$ and a stable filter%\todo[inline,size = \small]{Here and below, I have replaced $\lambda$ with $\sigma$ in ${\mathfrak F}$.} 
\begequ
\lab{f}
{\mathfrak F}(\calp)={\sigma \over \calp + \sigma}
\endequ
with $\sigma>0$, the following relation holds
$$
{\mathfrak F}[x \dot u] =  x \; \calp {\mathfrak F}[u]-{1 \over \calp + \sigma}[\dot x \, \calp {\mathfrak F}[u]].
$$
\end{lemma}
\section{Regression Equations for Parameter Identification}
\lab{sec5}
%%%%%%%%%%%%%%%5
%
In this section, proceeding from the key identity \eqref{keyide0}, we construct {\em two scalar} NLPREs that will be used in the next section for the estimation of the power coefficient parameters $c$. Before presenting this result the following remarks are in order. \\

\noindent {\bf R1} The interest of considering two, instead of just one, NLPRE stems from the fact that  the {\em excitation properties} of the matrix regressor---which gathers more information about the system---are better than the ones of the simple scalar one.\\

\noindent {\bf R2}  As mentioned in the Introduction, the parameterization that we obtain is, unfortunately, {\em nonlinear}. Since the classical gradient and least-squares estimators are unable to deal with NLPRE,  it is necessary to use a DREM-based parameter estimator like the ones presented in \cite{ORTetaltac21,WANetal}. This task is carried out in Section \ref{sec7} using an extension of the LS+DREM estimator first reported in \cite{PYRetal}.  \\

To present the main result of this section, namely the derivation of a NLPRE for the system to be used by the parameter estimator, requires to impose either  {\bf Assumption A3} or {\bf Assumption A2'} and {\bf Assumption A3'}. As explained in Subsection \ref{subsec42} the need for these assumptions stems from the fact that we need to remove the additive term for the system dynamics to generate the NLPRE. See the discussion about the need to remove this additive term in Section \ref{sec9}.

Again, as done in previous sections, because of their strong similarity we present in detail the case of scenario {\bf S1} and present without proof the one of scenario {\bf S2}.

\begin{proposition}
\lab{pro1}\em
Consider the system dynamics \eqref{mecdyn},  \eqref{tm} verifying  {\bf Assumptions A1}-{\bf A3}. There exist {\em measurable} signals $y(t)\in \rea^2$ and $\phi(t) \in \rea^{2 \times 4}$ such that the NLPRE
\begequ
\lab{nlpre}
y=\phi \; \calg (\theta) +\varepsilon_t
\endequ
holds, where $\varepsilon_t$ is an exponentially decaying term\footnote{This term stems from the effect of the initial conditions of the various filters. Following standard practice we disregard them in the sequel---see, however, the simulations in Section \ref{sec8}.} and the mapping  $ \calg :\rea^3 \to \rea^4$ is given by\footnote{We underscore the fact that $z(0)$ is known.}
\begequ
\lab{gthe}
 \calg (\theta):=e^{-\theta_3 z(0)}\begmat{\theta_1 \\  \theta_2 \\ \theta_1 \theta_3\\  \theta_2\theta_3}=\begmat{\calg_1 \\ \calg_2 \\ \calg_3 \\ \calg_4}.
\endequ
\end{proposition}

\begin{proof}
We will derive separately the two components of the NLPRE \eqref{nlpre}. Hence, we define 
$$
y = \begmat{y_1 \\ y_2},\; \phi=  \begmat{\phi_1^\top \\\phi_2^\top},
$$ 
and write the components  of the NLPRE as 
$$
y_i=\phi_i^\top \calg (\theta),\;i=1,2.
$$

\noindent {\em (First component)} Multiplying \eqref{keyide0} by $e^{-\theta_3 z(0)}$ and rearranging terms we get
 \begequ		
\dot z  =  \begmat{-z^4 & z^3 & - \xi_1\dot z &  {-}\xi_2\dot z} \calg (\theta),
\endequ
where we have used the fact that, due to {\bf Assumptions A3}, the signal $\tau_d$ defined in \eqref{taud} is zero. Applying the filter ${\mathfrak F}(\calp)$ defined in \eqref{f} we get
\begalis{
\calp {\mathfrak F}[z] & =- \calg_1  {\mathfrak F}[z^4]  + \calg_2  F {(\calp)}[ z^3]- \calg_3 {\mathfrak F}[\xi_1\dot{z}]  {-}\calg_4 {\mathfrak F}[\xi_2\dot{z}]\\
& = -\calg_1 {\mathfrak F}[z^4]  + \calg_2   F[ z^3] - \calg_3 \Big\{\xi_1 \calp {\mathfrak F}[{z}]-{1 \over \calp + \sigma}[\dot \xi_1\calp {\mathfrak F}[z]]\Big\} \\
&  {-} \calg_4\Big\{\xi_2 \calp {\mathfrak F}[{z}]-{1 \over \calp + \sigma}[\dot \xi_2\calp {\mathfrak F}[z]]\Big\}\\
& =- \calg_1 {\mathfrak F}[z^4]  + \calg_2   \calp F[ z^3] - \calg_3 \Big\{\xi_1 \calp {\mathfrak F}[{z}]+{1 \over \calp + \sigma}[z^4\calp {\mathfrak F}[z]]\Big\} \\
& +\calg_4\Big\{\xi_2 \calp {\mathfrak F}[{z}]-{1 \over \calp + \sigma}[z^3\calp {\mathfrak F}[z]]\Big\}\\
& = \phi_{ {1}}^\top \calg (\theta),
}	
where we applied the Swapping Lemma \ref{lem3} to get the second identity, \eqref{dotxi} for the third one and we defined the regressor vector
$$
 \phi_1:= \begmat{- {\mathfrak F}[z^4]  \\   { {\mathfrak F}[ z^3]} \\-\xi_1 \calp {\mathfrak F}[{z}]-{1 \over \calp + \sigma}[z^4\calp {\mathfrak F}[z]] \\  {-}\xi_2 \calp {\mathfrak F}[{z}] {+}{1 \over \calp + \sigma}[z^3\calp {\mathfrak F}[z]]}
 $$
for the last one. Identifying 
\begalis{
y_1&:=\calp {\mathfrak F}[ z],
}
completes the proof for the first component of the NLPRE.\\

\noindent {\em (Second component)} Multiplying \eqref{keyide0} by $z^{-3}$  we get
\begequ
\lab{seccom}
 e^{\theta_3 z(0)}{\dot z  \over z^3} + \theta_1\theta_3{\dot z  \over z^3}\xi_1  + \theta_2\theta_3{\dot z  \over z^3}\xi_2 + \theta_1 z - \theta_2 =0,
\endequ
Define the signal
$$
\xi_3:=-{1  \over 2z^2},
$$
and notice that 
$$	
\dot \xi_3  = {\dot z  \over z^3}.
$$
Replacing the equation above in \eqref{seccom}  and rearranging terms yields
$$
 e^{\theta_3 z(0)} \dot \xi_3  = - \theta_1 z + \theta_2- \theta_1\theta_3\dot \xi_3   \xi_1  - \theta_2\theta_3\dot \xi_3   \xi_2.
$$
Multiplying by $e^{\theta_3 z(0)}$ we get
$$
 \dot \xi_3  = \begmat{-z & 1 & -  \xi_1 \dot \xi_3  & -    \xi_2\dot \xi_3}\calg(\theta) .
$$
Applying the filter ${\mathfrak F}(\calp)$  we get
$$
  \calp {\mathfrak F}[\xi_3 ]  = \begmat{- {\mathfrak F} [z] &  {\mathfrak F} [1] & - {\mathfrak F} [  \xi_1 \dot \xi_3 ] & -  {\mathfrak F} [  \xi_2\dot \xi_3 ]}\calg(\theta).
  $$
The proof is completed proceeding as done for the first component and defining
\begalis{
 \phi_2&:= \begmat{-{\mathfrak F}[z]  \\    {\mathfrak F}[ 1] \\ \xi_1 \calp {\mathfrak F}[{1  \over 2z^2}] {+}{1 \over \calp + \sigma}[z^4\calp {\mathfrak F}[{1  \over 2z^2}]] \\ \xi_2 \calp {\mathfrak F}[{1  \over 2z^2}] {-}{1 \over \calp + \sigma}[z^3\calp {\mathfrak F}[{1  \over 2z^2}]]}\\
 y_2&:=\calp {\mathfrak F}[ \xi_3].
}
\end{proof}

The corollary below pertains to scenario {\bf S2} where we proceed in the derivations from \eqref{keyide1} instead of \eqref{keyide0}. The proof of the corollary follows {\em verbatim} the one given above---hence it is omitted for brevity.

\begin{corollary}
\lab{cor3}\em
Consider the system dynamics \eqref{mecdyn},  \eqref{tm} verifying  {\bf Assumptions A1}, {\bf A2'} and {\bf A3'}. There exists {\em measurable} signals $\bar y(t)\in \rea^2$ and $\bar \phi(t) \in \rea^{2 \times 4}$ such that the NLPRE
\begequ
\lab{nlpre0}
\bar y=\bar \phi \; \calg (\bar\theta)
\endequ
holds, where the mapping  $ \calg :\rea^3 \to \rea^4$ is given by\footnote{We underscore the fact that $z(0)$ is known.}
$$
 \calg (\bar \theta):=e^{-\bar \theta_3 z(0)}\begmat{\bar \theta_1 \\ \bar \theta_2 \\ \bar \theta_1 \bar \theta_3\\ \bar \theta_2 \bar \theta_3}.
$$
\end{corollary}
%
%%%%%%%%%%%%
\section{Construction of a Strictly Monotonic Mapping}
\lab{sec6}
%%%%%%%%%%%
%
In view of the similarity between \eqref{nlpre} and \eqref{nlpre0} in the remaining part of the paper we concentrate only on the first one, which corresponds to the scenario {\bf S1}. 

As discussed in the Introduction, although it is possible to define an {\em extended} parameter vector $\theta_a:=\calg(\theta)$ to obtain a LRE to which any standard estimator is applicable, the lack of excitation of the regressor in the present problem stymies this approach. This fact is illustrated in the simulations of Section \ref{sec8}. Moreover, overparametrization suffers from  well-known shortcomings  that include a performance degradation, {\em e.g.}, slower convergence, due to the need of a search in a larger parameter space. See \cite{ORTetalaut21} for a thorough discussion on this issue.

Consequently, we have to deal with the NLPRE and try to estimate only the vector $\theta$. In Section \ref{sec7} we {\em extend}  the LS+DREM estimator proposed in \cite{PYRetal}, which is applicable for LRE, to deal with NLPRE of the form   \eqref{nlpre}. However, it is required that the mapping of the NLPRE satisfies a {\em monotonicity} property, which is not verified by $\calg(\theta)$. Therefore, in this section we construct a new mapping verifying the required monotonicty condition.

\begin{lemma}
\lab{lem4}\em
Consider the mapping $\calg(\theta)$ given in  \eqref{gthe}. Define the injective mapping
\begequ
\lab{invmap}
\theta=\cald(\eta):=\begmat{e^{\eta_3 z(0)}\eta_1 \\ e^{\eta_3 z(0)} \eta_2 \\  \eta_3}, 
\endequ
whose inverse is given as
\begequ
\lab{injmap}
\eta=\begmat{e^{-\theta_3 z(0)}\theta_1 \\ e^{-\theta_3 z(0)} \theta_2 \\  \theta_3}=\begmat{\eta_1 \\ \eta_2 \\ \eta_3 }.
\endequ
The composed mapping $\calw:\rea^3 \to \rea^4$  defined as
\begali{
\calw(\eta)& :=\calg(\cald(\eta)).
\lab{psic}
}
satisfies the {\em linear matrix inequality}
\begequ
\lab{demcon1}
T\nabla \calw(\eta) + [\nabla \calw(\eta)]^\top T^\top \geq \rho I_3 >0,
\end{equation}
for the matrix
$$
T:=\begmat{\alpha & 0 & 0 & 0 \\ 0 & \alpha & 0 & 0 \\  0 & 0 & 0 & 1},
$$
provided
\begequ
\lab{conalp}
\alpha > {1 \over 4}{\eta_3^2 \over \eta_2}.
\endequ
\end{lemma}

\begin{proof}
The mapping $\calw(\eta)$  is given by
$$
\calw(\eta)=\begmat{\eta_1 \\ \eta_2 \\ \eta_1 \eta_3 \\ \eta_2 \eta_3},
$$
whose Jacobian is
$$
\nabla \calw(\eta)=\begmat{1 & 0 & 0\\ 0 & 1 & 0 \\  \eta_3 & 0 & \eta_1\\0 &  \eta_3 & \eta_2}.
$$
The symmetric part of the matrix $T \nabla\calw(\eta)$ takes the form
$$
T  {\nabla}\calw(\eta)+ [\nabla \calw(\eta)]^\top T^\top =\begmat{2 \alpha & 0 & 0\\ 0 & 2 \alpha & \eta_3 \\ 0 &  \eta_3 & 2 \eta_2}.
$$
Some simple calculations with the Schur complement prove that this matrix is positive definite if and only if \eqref{conalp} is satisfied
\end{proof}
\subsubsection*{Discussion}  
\noindent {\bf D4} In \cite[Proposition 1]{ORTetalaut21} it is shown that \eqref{demcon1} ensures  the mapping $T \calw(\eta) $ is {\em strictly monotonic} \cite{PAVetal}. That is, it satisfies
\begequ
\lab{monpro}
(a-b)^\top \left[T\calw(a) -T \calw(b)\right] \geq \rho |a-b|^2,\;\forall a,b \in \rea^3,\; a \neq b.
\endequ 
This is the fundamental property that is required by the LS+DREM estimator used in the next section.\\ 

\noindent {\bf D5} Using  \eqref{the} and  \eqref{injmap}, the condition  \eqref{conalp} can be expressed in terms of the physical parameters as follows
$$
\alpha > {1 \over 4}{\eta_3^2 \over \eta_2}= {1 \over 4}\theta_3^2 {e^{-\theta_3 z(0)}\over \theta_2}= {J \over 4 \kappa v_w} c_3^2{ e^{{-c_3 v_w \over \omega(0)}}\over c_1 c_2}.
$$
Consequently, to satisfy this condition some prior knowledge on the parameters $c$  is required. Specifically, it is necessary to know lower bounds on $c_1$ and $c_2$ and an upperbound on $c_3$. From the physical viewpoint this is not a restrictive assumption.
%
%%%%%%%%%%%%
\section{An  Estimator of the Power Coefficient Parameters $c$}
\lab{sec7}
%%%%%%%%%%%
%
In this section we present the main result of the paper, that is, an estimator of the power coefficient parameters $c$ that achieves {\em global, exponential} convergence of the parameter error imposing  the following extremely weak {\em IE} assumption \cite{KRERIE,TAObook} of the regressor  $\phi$ of the NLPRE \eqref{nlpre}.  \\

\noindent {\bf Assumption A4} [Excitation] The regressor matrix $\phi$ of the NLPRE \eqref{nlpre} is IE. That is, there exists constants $C_c>0$ and $t_c>0$ such that
 \begalis{
	&\int_0^{t_c} \phi^\top(s) \phi(s)  ds \ge C_c I_4.
}

\begin{proposition}\em
\lab{pro2}
Consider the NLPRE (\ref{nlpre}) that, using \eqref{psic}, we rewrite as
\begequ
\lab{nlpre1}
y=\phi  \calw(\eta),
\endequ
with $\phi$ verifying   {\bf Assumption A4}. Define the the LS+DREM interlaced estimator
\begsubequ
\lab{intestt}
\begali{
\lab{thegt}
\dot{\hat \calw} & =\gamma_\calw F \phi^\top (y-\phi \hat\calw),\; \hat\calw(0)=\calw_{0} \in \rea_+^4\\
\lab{phit}
			\dot {F}& =  -\gamma_\calw F \phi^\top  \phi  F,\; F(0)={1 \over f_0} I_4 \\
\lab{thet}
\dot{\hat \eta} & =\Gamma \Delta T  [Y -\Delta \calw(\hat\eta) ],\; \hat\eta(0)=\eta_{0} \in \rea^3_+,
}
\endsubequ		
with tuning gains the scalars $\gamma_\calw>0$, $f_0>0$ and the positive definite matrix $\Gamma \in \rea^{3 \times 3}$,  and we defined the signals
\begsubequ
\lab{aydelt}
\begali{
\lab{delt}
\Delta & :=\det\{I_4-f_0F\}\\
\lab{yt}
Y & := \adj\{I_4- f_0F\} (\hat\calw -  f_0F \calw_{0}),
}
\endsubequ
where $ \adj\{\cdot\}$ denotes the adjugate matrix. Define the estimated power coefficient parameters
\begequ
\lab{hatc}
\hat c=\begmat{{J\over \kappa v_w} e^{\hat\eta_3 z(0)}\hat \eta_1    \\ {\hat \eta_2 \over \hat \eta_1} \\ \hat {\eta_3}}.
\endequ
Then, for all initial conditions $\calw_{0} \in \rea_+^4$ and $\eta_{0} \in \rea^3_+$, we have that the estimation errors of the power coefficient  parameters  verify
\begequ
\lab{parcon}
\lim_{t \to \infty}|\tilde {c}(t)|=0,\;(exp),
\endequ
with all signals {\em bounded}. 
\end{proposition}

\begin{proof}
We will prove that the parameter error  $\tilde \eta :=\hat \eta -\eta$ converges to zero exponentially fast. Hence, invoking \eqref{hatc},  we conclude \eqref{parcon}.

Replacing \eqref{nlpre1} in \eqref{thegt} yields the error dynamics for the first LS estimator
$$
\dot{\tilde \calw}   = -\gamma_\calw F \phi^\top \phi \tilde\calw,
$$
where $\tilde \calw :=\hat \calw -\calw(\eta)$. Now, from the well-known fact, {\em e.g.}, \cite[Theorem 4.3.4]{IOASUNbook}, that LS estimators satisfy,
$$
{d \over dt}(F^{-1}\tilde \calw)  =0
$$
we have 
$$
\tilde \calw(t)  =f_0F(t)\tilde\calw(0),
$$
which may be rewritten as the {\em extended} NLPRE
\begali{
\lab{keyide}
(I_4-f_0F)\calw(\eta) &=\hat \calw - f_0F \calw_{0}.
}
Following the DREM procedure we multiply \eqref{keyide} by   $\adj\{I_4-f_0F\}$ to get the following  NLPRE
\begequ
\lab{ydel}
Y= \Delta\calw(\eta),
\endequ
where we used \eqref{delt} and \eqref{yt}.  We underscore the fact that the regressor $\Delta $ is a {\em scalar}. Replacing \eqref{ydel} in \eqref{thet} yields the error dynamics for the vector  $\tilde\eta $ as
\begalis{
  {\dot{\tilde{\eta}}} & =-{ \Gamma \Delta^2}T[ \calw(\hat\eta) - \calw(\eta)] .
}
To analyse its stability define the Lyapunov function candidate 
$$
U(\tilde \eta) := \frac{1}{2} \tilde \eta^\top \Gamma^{-1} \tilde \eta,$$ 
whose time derivative yields
\begalis{
\dot U & =  - \Delta^2 [ \hat \eta - \eta]^\top  T{[ \calw(\hat \eta) - \calw(\eta)]} \\
& \leq  - \Delta^2 \rho | \tilde \eta|^2  \\
& \leq  - {2\rho {\Delta}^2\over \lambda_m\{\Gamma\}} U,
}
where we invoked the strong monotonicity property \eqref{monpro} of $T \calw(\eta)$ to get the first bound and $\lambda_m\{\cdot\}$ denotes the minimum eigenvalue. To complete the proof, we invoke the Comparison Lemma \cite[Lemma 3.4]{KHAbook} that yields the bound
$$
U(t+t_c) \leq e^{- {2\rho \over \lambda_m\{\Gamma\}} \int_t^{t+t_c} \Delta^2(s)ds}U(t),
$$
which ensures 
$$
\lim_{t \to \infty}\tilde \eta(t)=0\;(exp),
$$
if  $\Delta$ is persistently exciting. The latter condition follows from the assumption that $\phi(t)$ is IE and \cite[Lemma 3.5]{TAObook}, which ensures the following
$$
\phi(t) \in IE\;\Rightarrow\;\Delta(t) >0,\;\forall t \geq t_0+t_c.
$$
Consequently, $\Delta$ is persistently exciting.
  \end{proof}
%
%%%%%%%%%%%%
\section{Simulation Results}
\lab{sec8}
%%%%%%%%%%%%%%%
%
In this section we present some simulation evidence of the proposed estimator for the system operating under the conditions of scenario {\bf S1}, as well as some robustness results for it.

The parameters $\kappa_i$ of the power coefficient function in \eqref{heier} were taken from \cite{HEIbook} and are: $\kappa_1=0.5$, $\kappa_2=116$, $\kappa_5=5$, $\kappa_6=21$ and $\kappa_7=0.035$. This parameter selection yields the values of $c_i$ shown in Table \ref{table1}. The same table indicates the parameters of the mechanical dynamics and the wind value used in the simulation. The values of the estimator parameters used in the simulation are given in Table~\ref{table2}. The initial values for the parameters $\eta$ are chosen such that they correspond to the power curve equation given in \cite{KUMCHA}, which is different from the one used in \cite{HEIbook}; see Fig.~\ref{fig:cp2} for comparison.

Given the very weak excitation conditions of the system it is necessary to remove the effect of the initial conditions of the various filters in the generation of the NLPRE \eqref{nlpre}, that is, to set $\varepsilon_t=0$. We adopt then the standard procedure of adding the terms due to the unknown initial conditions to the regressor, see \cite[Section V]{LOZZHA}. 

To validate the approach, two simulation cases were carried out which are described in the following.

\subsection{$T_e=0$ and constant $v_w$ }
In this first simulation case, the conditions under which the estimator is tested coincide with those in which the estimator is proven to have zero convergence in the estimation errors. Namely, the wind is maintained at constant speed (see Table \ref{table1}) and $T_e=0$. The set of results are shown in  Fig. \ref{case1}. In Fig. \ref{fig:zdz}, the plot $z$ \textit{vs} $\dot z$ from \eqref{zdyn} is depicted. As seen from the figure, the $z$-dynamics has a stable equilibrium point at $z=0.14$, which is the value to which $z(t)$ converges in the simulation Note that, $z=0$ is also an equilibrium point, but this is ruled out by the assumption that $\omega(t)>0$.  

The estimation errors $\tilde {c}_i$ during the simulation have been normalized and are shown in Fig. \ref{errors}. Besides, Fig. \ref{fig:cp1} shows the actual $\mathcal{C}_p$ used in simulation and the estimated $\hat {\mathcal{C}}_p$, which is reconstructed from the estimates $\hat{c}_i$. As seen from the figure, both plots overlap.  Likewise, Fig. \ref{fig:cp2} shows the estimation $\hat{\mathcal{C}_p}$ at time $t=0$, that is,  the function due to initial estimates $\hat c_i(0)$, and that at the end of simulation. 

It is argued in Section \ref{sec6} that, due to the poor excitation conditions, the possibility of using a standard (gradient or LS) estimator for an overparamterized version of the NLPRE \eqref{nlpre} is ruled out. To illustrate this fact notice that the first set of equations of the proposed interlaced LS+DREM estimator, that is equations \eqref{thegt}, are just a standard LS algorithm for the overparameterized NLRE. It is well-known \cite{IOASUNbook,SASBODbook} that a {\em necessary} condition for convergence of the LS estimator is that the (covariance) matrix $F$ converges to zero. In Fig. \ref{figf} we show that this is not the case for our simulation. On the other hand, in Fig. \ref{delta} we show that, as predicted by the theory,  $\Delta$ converges to a constant non-zero value, hence it is persistently exciting. 
 
\subsection{Robustness}
To test the estimator's robustness, we consider the case when measurements of wind and rotor velocities have additive noise. Specifically, for simulations we consider uniform noise with a magnitude of $0.3$ m/s for the wind's velocity and $0.5$ rad/s for the rotor's velocity. Due to the insufficient excitation of the regressor $\phi$, the noises yield bias in the estimates. However, the reconstructed power curve allows to accurately estimate the point where its maximum value is attained. These facts are illustrated in Fig.~\ref{case_noise}. 

\begin{table}[t!]
	\centering
	\begin{center}
		\footnotesize
				\begin{tabular} {p{2.7cm}p{1.9cm}}
				\toprule
				\midrule
				\multicolumn{1}{c}{Parameter} & \multicolumn{1}{c}{Magnitude}\\%&\multicolumn{2}{c}{Storage} \\
				\midrule
					Wind speed $v_w$ & \multicolumn{1}{c}{$9\;[\mathrm{\frac{m}{s}}]$}  \\
					Air density $\rho$ & \multicolumn{1}{c}{$1.225\;[\mathrm{\frac{kg}{m^3}}]$} \\
				Blade length $r$  &  \multicolumn{1}{c}{$1.84\;[\mathrm{m}]$}  \\
				Inertia  $J$ &\multicolumn{1}{c}{$7.856\; [\mathrm{kg\cdot m^2}]$}  \\
				Pitch angle $\beta$ &  \multicolumn{1}{c}{$0\; [\mathrm{rad}]$}\\
			     $c_1$ &   \multicolumn{1}{c}{$65.74$} \\
				 $c_2$ & \multicolumn{1}{c}{$0.144$} \\
				 $c_3$& \multicolumn{1}{c}{$11.41$} \\
				\midrule
				 \bottomrule
		\end{tabular}
	\end{center}
	\caption{Parameter values of \eqref{mecdyn} and \eqref{cp}.}
\label{table1}
\end{table}

\begin{table}[t!]
	\centering
	\begin{center}
		\footnotesize
		\begin{tabular} {p{1.7cm}p{1.9cm}}
			\toprule
			\midrule
			\multicolumn{1}{c}{Parameter} & \multicolumn{1}{c}{Magnitude}\\%&\multicolumn{2}{c}{Storage} \\
			\midrule
			$\omega(0)$ & \multicolumn{1}{c}{$10$} \\
			$\sigma$  & \multicolumn{1}{c}{$1$}  \\
			$\gamma_\calw$ & \multicolumn{1}{c}{$100$}\\
			$\Gamma$ & \multicolumn{1}{c}{$\operatorname{diag}\!(50,\;50,\;500)$}\\
			$f_0$  & \multicolumn{1}{c}{$1$} \\
			$\calw_0$  & \multicolumn{1}{c}{$0$}\\
			\midrule
			\bottomrule
		\end{tabular}
	\end{center}
	\caption{Parameter values for the estimator \eqref{f} and \eqref{intestt}.}
	\label{table2}
\end{table}

\begin{figure}[t!]
	\centering
	\subfloat[$z$ \textit{vs} $\dot z$]{\includegraphics[width = 0.5\linewidth]{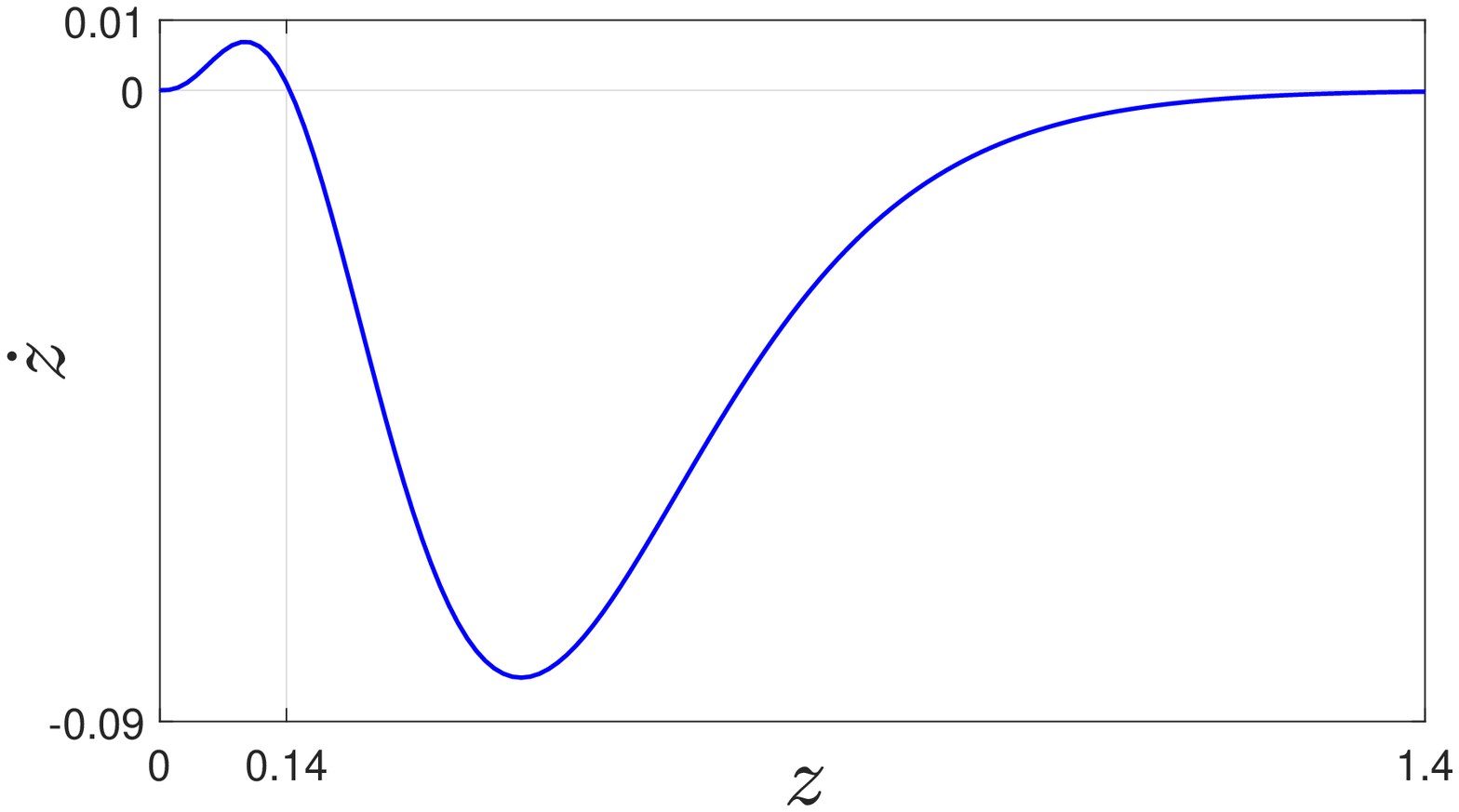}\label{fig:zdz}}\hfill
	\subfloat[Normalized estimation errors $\tilde {c}_i$]{\includegraphics[width = 0.5\linewidth]{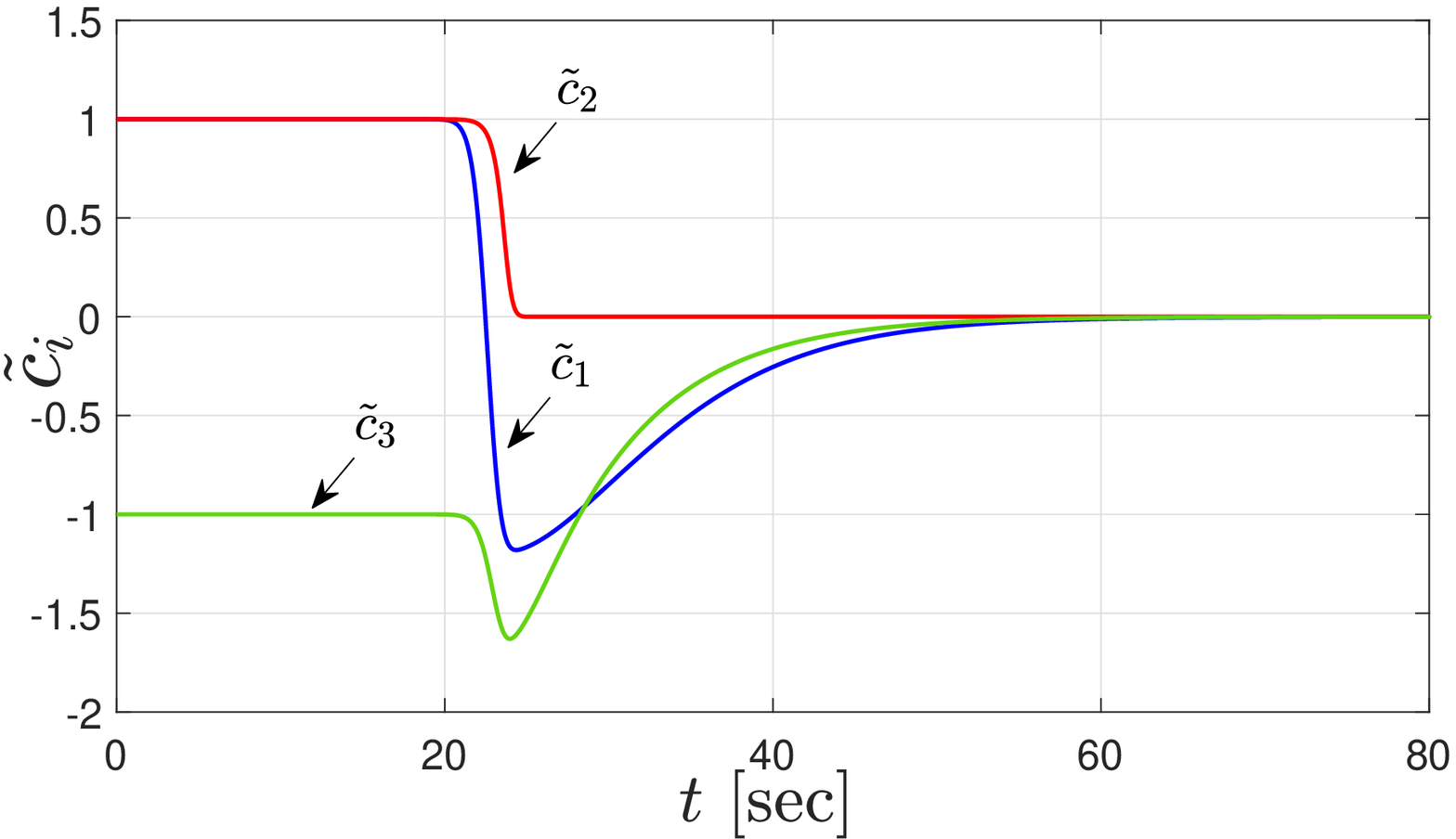}\label{errors}}\\
	\subfloat[$\mathcal{C}_p$ and $\hat{\mathcal{C}}_p$ ]{\includegraphics[width = 0.5\linewidth]{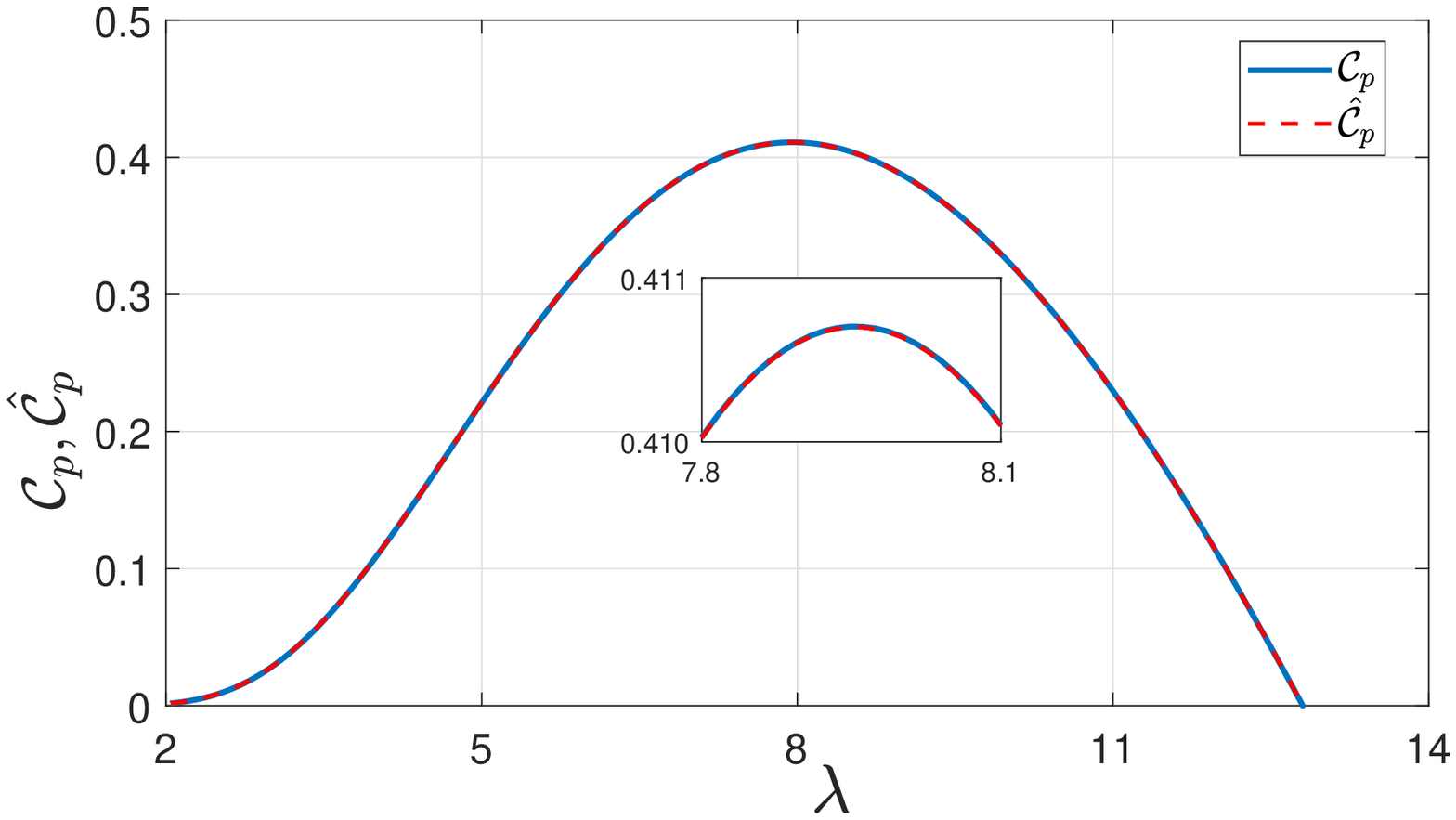}\label{fig:cp1}}\hfill 
	\subfloat[$\hat{\mathcal{C}}_p$ at $t=0$ and $\hat{\mathcal{C}}_p$ at $t=t_f$ (simulation final time)]{\includegraphics[width = 0.5\linewidth]{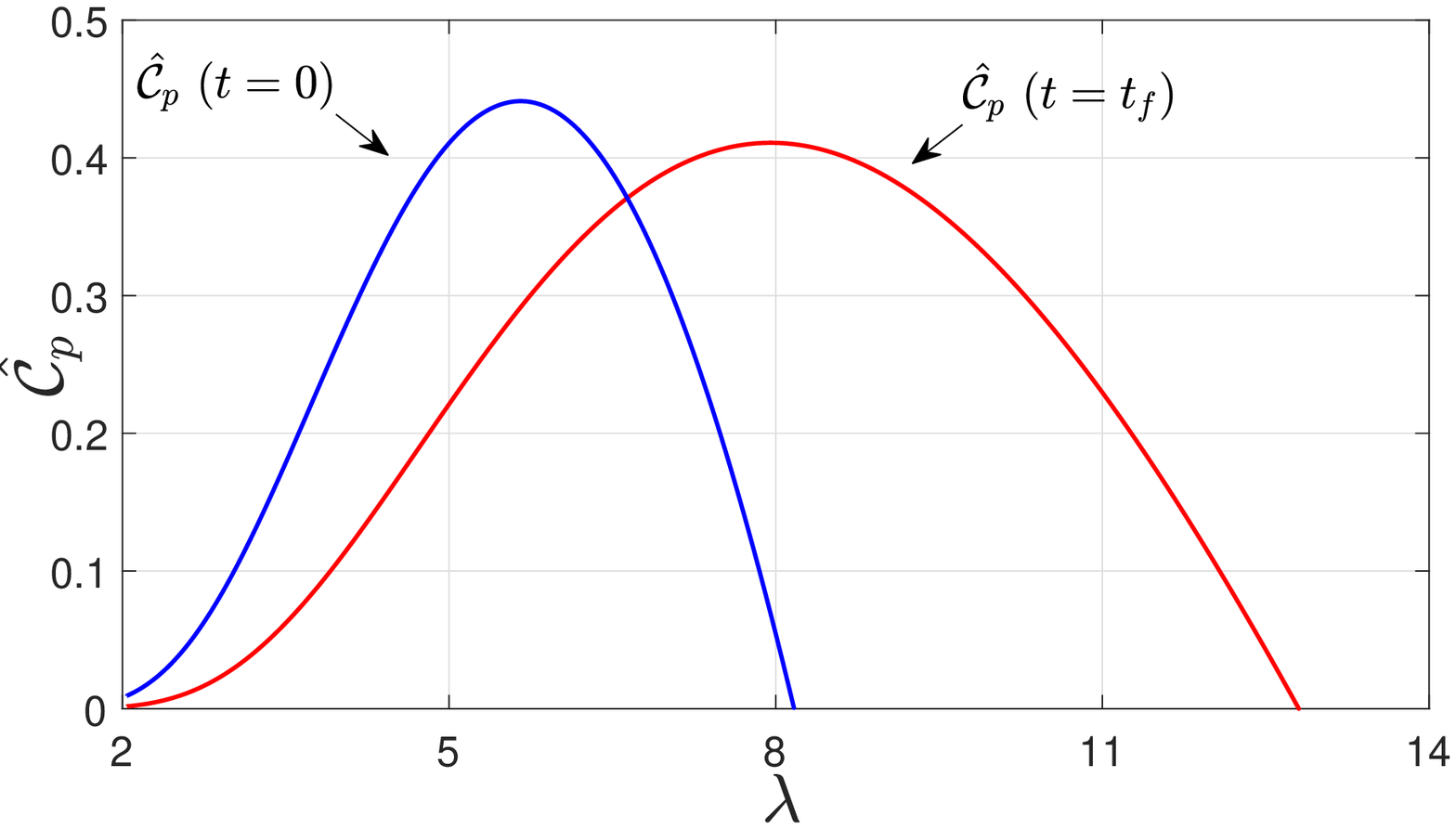}\label{fig:cp2}}\\
	\subfloat[$\Delta$]{\includegraphics[width = 0.49\linewidth]{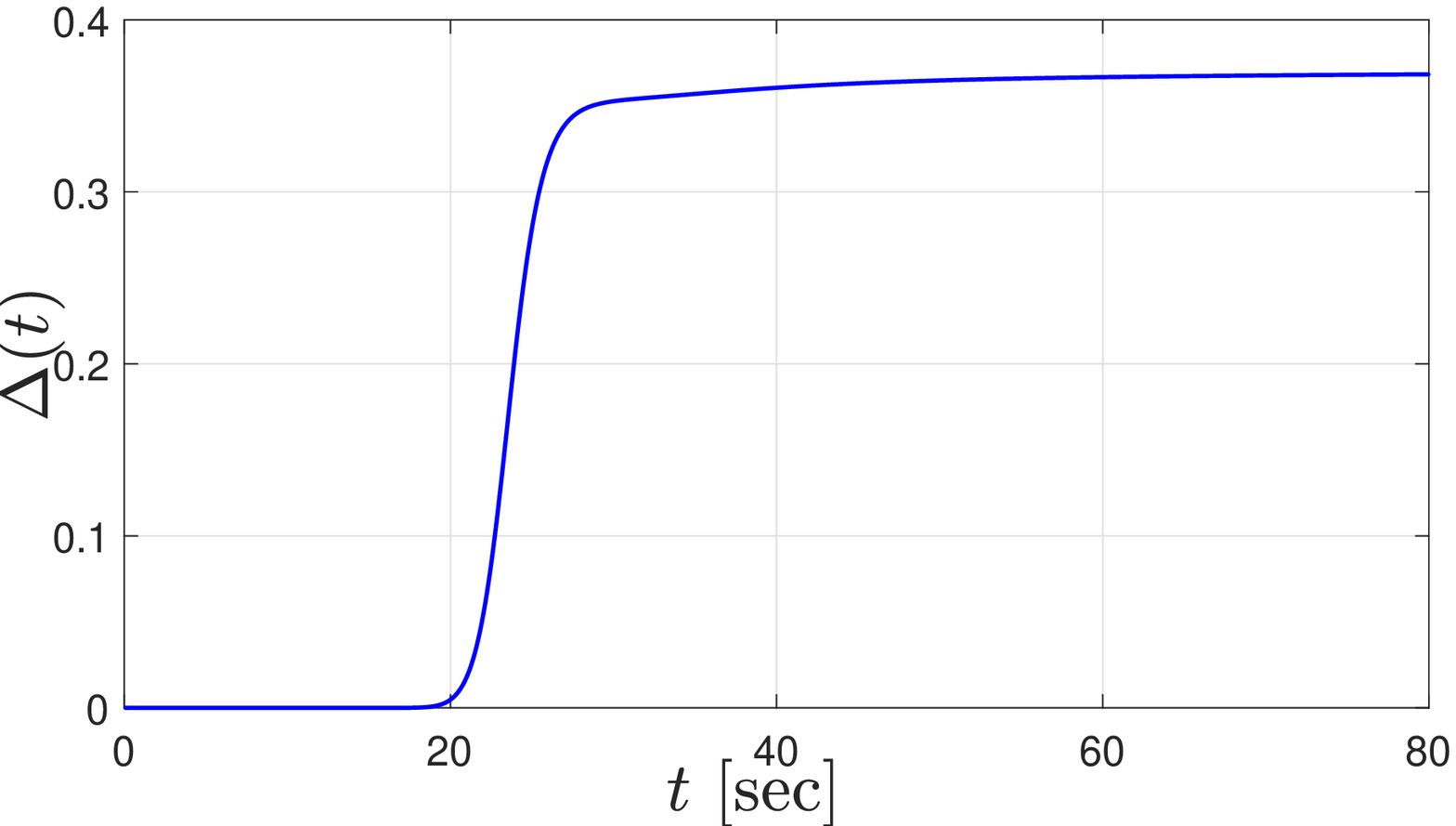}\label{delta}}\hfill
	\subfloat[$\lambda_M \{F(t)\}$]{\includegraphics[width = 0.49\linewidth]{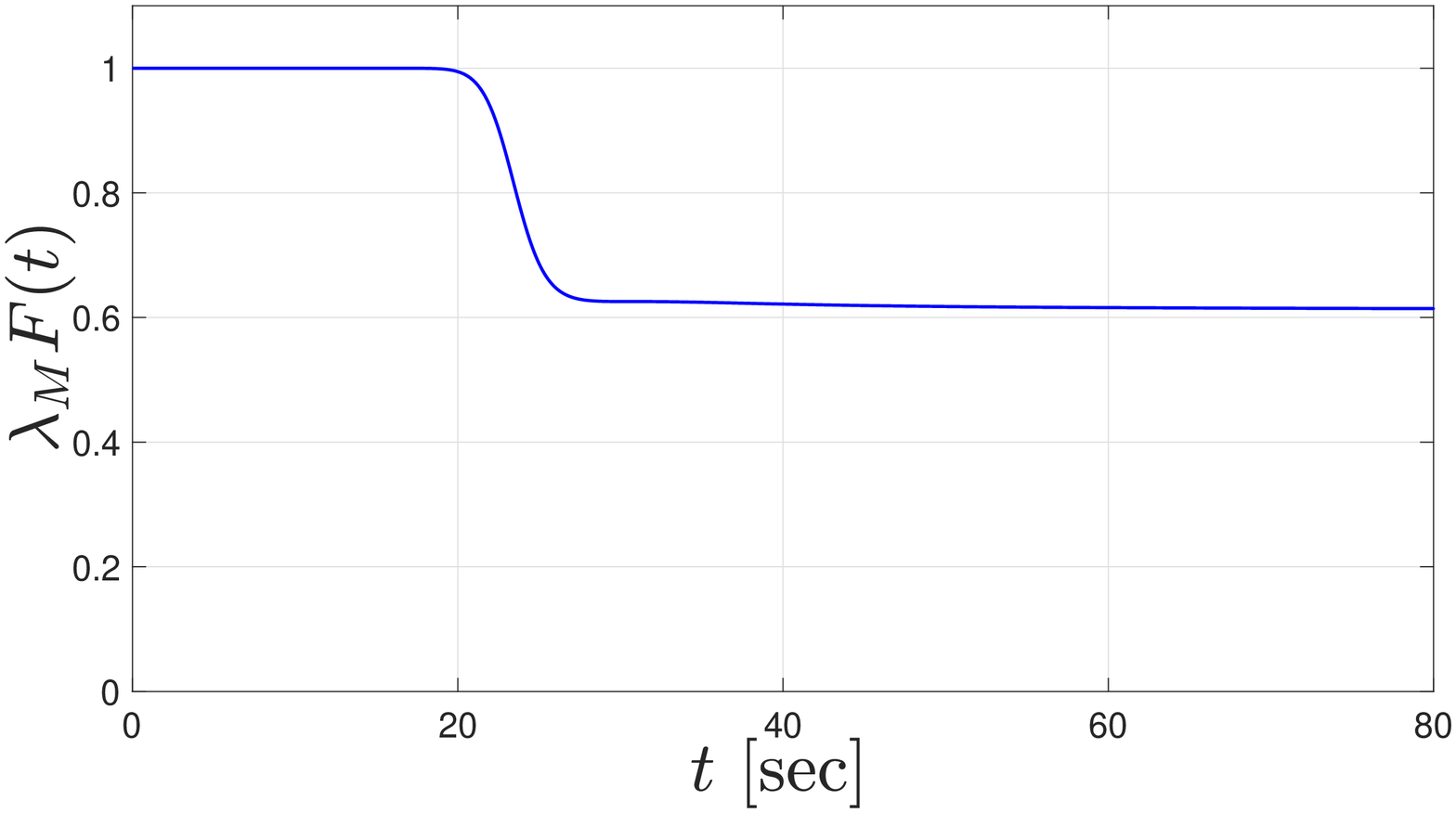}\label{figf}}
	\caption{Simulation results for $T_e=0$ and constant $v_w=9\;\mathrm{m/s}$.}
	\label{case1}
\end{figure}

\begin{figure}[t!]
	\centering
	\subfloat[Normalized estimation errors $\tilde {c}_i$]
		{\includegraphics[width = 0.45\linewidth]{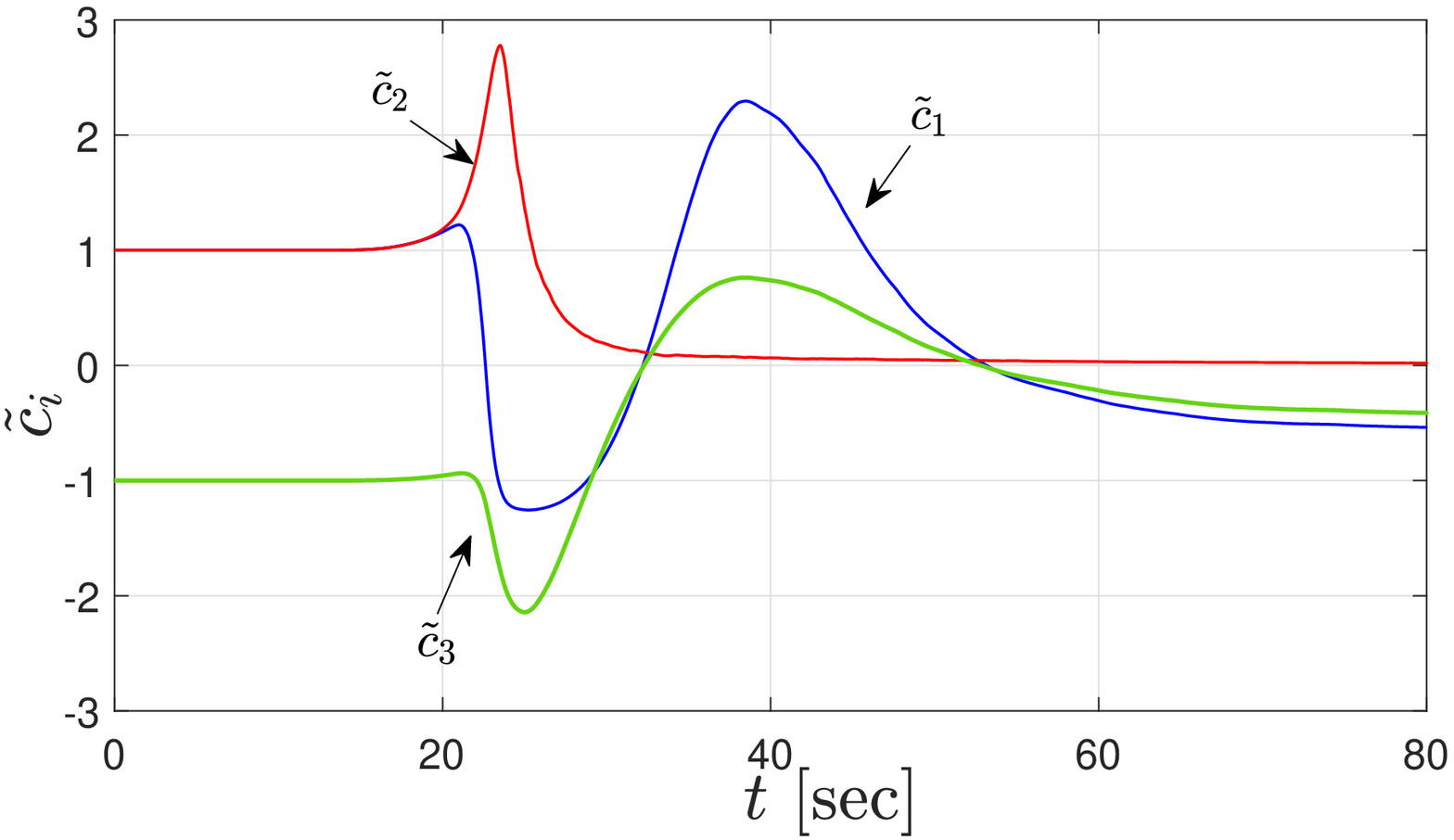}\label{errors_noised}}
	\hfill
	\subfloat[$\mathcal{C}_p$ and $\hat{\mathcal{C}}_p$ ]
		{\includegraphics[width = 0.45\linewidth]{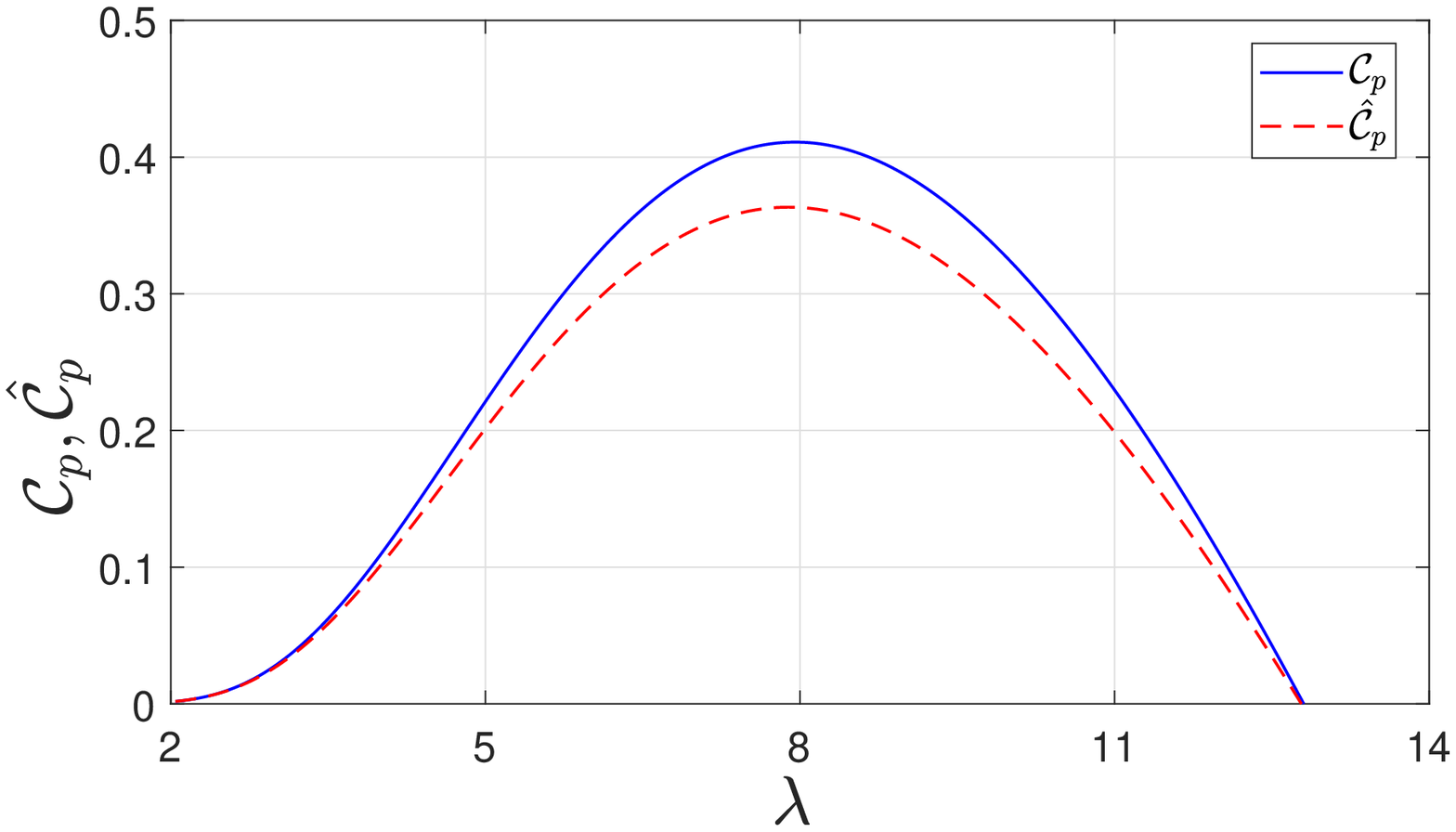}\label{fig:cp1_noised}}

	\caption{Simulation results for noisy measurements of the wind's velocity $v_w$ and the rotor's velocity $\omega$.}
	\label{case_noise}
\end{figure}

%%%%%%%%%%%
\section{Concluding Remarks and Future Research}
\lab{sec9}
%%%%%%%%%%%%
%
In this paper we have provided the first solution to the challenging problem of on-line estimation of a windmill power coefficient, which corresponds to a nonlinear, nonlinearly parameterized, underexcited system. This problem is of particular interest to solve the MPPT task using model-based strategies.

The result presumes that the windmill turbine operates in either an off-grid mode (scenario {\bf S1}) or with the generator torque defined in \eqref{refte} (scenario {\bf S2}). It is possible to show that if we remove {\bf Assumptions A3} in the scenario {\bf S1}, that is, if $T_e \neq 0$, then the NLPRE \eqref{nlpre} is perturbed by an additive term taking the form
$$
y=\phi \; \calg (\theta) +d(t,T_e),
$$
where $d$ may be viewed as a {\em disturbance}, which is $\calo(T_e)$.\footnote{{$d(t,T_e)$ is  $\calo(T_e)$, refered in the literature as``big o of $T_e$",} if and only if  $|d(t,T_e)| \leq C T_e$ with $C$ a constant independent of $t$ and $T_e$.} More precisely, some lengthy, but straightforward calculations, show that $d$ is defined as 
\begalis{
d:=e^{-\theta_3 z(0)}\begmat{{\mathfrak F}[\tau_d]\\ {\mathfrak F}\Big[{ \tau_d \over z^3}\Big]},
}
where $\tau_d$ is defined in \eqref{taud}. Unfortunately, the presence of such a disturbance term in the NLPRE hampers the application of the proposed parameter estimator. On the other hand, the $\calo(T_e)$ property of the disturbance $d$ suggests that the overall performance of the estimator will not be seriously degraded for ``small values" of $T_e$---a fact that has been shown in simulations. This difficulty is partially overcome in scenario {\bf S2}, but at the prize of imposing the particular profile \eqref{refte} to the generator torque. It is unclear at this point if this might lead us to a {\em bona fide} on-line adaptive MPPT controller.

We are currently working on the {\em practical} implementation of the proposed identification strategy on a physical windmill turbine and we expect to be able to report our results in the near future.

A word of caution regarding the practical implementation of the estimator pertains to the boundedness of the signals $\xi_1$ and $\xi_2$. From \eqref{dotxi} it is clear that they are open loop integrations of sign-definite functions that converge to non-zero constant values. Therefore, as $t \to \infty$ they would grow unbounded. It is then necessary to run the estimator only for a finite time, something that is consistent with the experimental scenario where the rotor speed will quickly converge to a constant value providing no additional information for the parameter estimation. 

Another issue that has to be addressed in a practical scenario is to avoid the possibility of the estimated parameter $\hat \eta_1$ crossing through zero since, in that case, a singularity in the computation of $\hat c_2$ would appear---see \eqref{hatc}. This contingency can be easily avoided, without affecting the theoretical result, adding a projection to the estimator as it is routinely done in adaptive systems \cite{POMPRA}.

\end{document}